\documentclass[11pt]{article}             % Use with LaTeX2e
\usepackage{osid}                         %

\usepackage{amsmath,amssymb,amscd,amsthm,mathrsfs}
\usepackage{bbm}
\usepackage{comment} 
\usepackage{enumerate} 
\usepackage{color}

\theoremstyle{plain}% Theorem-like structures provided by amsthm.sty
\newtheorem{lemma}{Lemma}
\newtheorem{thm}{Theorem}

\newtheorem{corollary}{Corollary}
\newtheorem{proposition}{Proposition}

\theoremstyle{definition}
\newtheorem{defi}{Definition}
\newtheorem{example}{Example}
\newtheorem{remark}{Remark}

\renewcommand{\vec}[1]{{\rm vec}(#1)}

\newcounter{app}

\title{Understanding and Generalizing Unique Decompositions of Generators of Dynamical Semigroups}
\author{Frederik vom Ende%\thanks{Supported by ...} 
	\\[1mm]{\footnotesize\it Dahlem Center for Complex Quantum
Systems, Freie Universität Berlin, Arnimallee 14, 14195 Berlin, Germany
 \& {frederik.vom.ende@fu-berlin.de}}\\[2ex]
}

\AtBeginDocument{
  \label{CorrectFirstPageLabel}
  
}

\begin{document}

\maketitle
\begin{abstract}
We generalize the result of Gorini, Kossakowski, and Sudarshan [J{.} Math{.} Phys{.} 17:821, 1976] that every generator of a quantum-dynamical semigroup decomposes uniquely into 
a closed and a dissipative part, assuming the trace of both vanishes.
More precisely, we show that given any generator $L$ of a completely positive 
dynamical semigroup and any matrix $B$ there exists a unique matrix $K$ and 
a unique completely positive map $\Phi$ such that (i) $L=K(\cdot)+(\cdot)K^*
+\Phi$, (ii) the superoperator $\Phi(B^*(\cdot)B)$ has trace zero, and (iii) 
${\rm tr}(B^*K)$ is a real number.
The key to proving this is the relation between the trace of a completely 
positive map, the trace of its Kraus operators, and expectation values of its 
Choi matrix.
Moreover, we show that the above decomposition is orthogonal with respect to some $B$-weighted inner product.
\end{abstract}

\section{Introduction}
Over the last decade, advances in quantum engineering have surged interest in open quantum systems and dissipative dynamics.
Fields where this interest has shown include---but are not limited to---quantum thermodynamics \cite{MDAS19,SA20,Kosloff21,SCE22,GS23},
quantum control \cite{Roadmap2022,Kosloff21_2,HS21,OSID_thermal_res,PP23},
as well as many-body systems \cite{NR20} and quantum sensing \cite{SSG23}.
At the core of open systems theory is undoubtedly the Gorini-Kossakowski-Sudarshan-Lindblad (GKSL) \cite{GKS76,Lindblad76} equation---which, remarkably,
is still being investigated today from a mathematical \cite{CKKS21} as well as a modeling point-of-view \cite{Trushechkin21,Trushechkin22}.

The GKSL-equation is a first-order master equation that characterizes the dynamics of open quantum systems, assuming the system in question satisfies the Markovian approximation \cite{BreuPetr02}.
The GKSL-equation comprises a closed-system part (following the Liouville-von Neumann equation generated by some Hamiltonian $H$) as well as an open-system term $\bf\Gamma$ which models the system-environment interaction by means of so-called Lindblad operators $\{V_j\}_{j=1}^m$:
\begin{equation}\label{eq:gksl}
\dot\rho(t)=-i[H,\rho(t)]-\underbrace{\sum_{j=1}^m\Big( \frac12(V_j^*V_j\rho(t)+\rho(t)V_j^*V_j)-V_j\rho(t)V_j^* \Big)}_{{\bf\Gamma}(\rho(t))}\,,\quad\rho(0)=\rho_0\,.
\end{equation}
We refer to Eq.~\eqref{eq:liewedge_cptp} for the precise statement.
The operator $L:=-i[H,\cdot]-{\bf\Gamma}$ is sometimes called GKSL-generator.
% (see \cite{ChruPas17} for a historic overview)

Already in their seminal paper \cite{GKS76}, Gorini, Kossakoswki, and Sudarshan raised the question under what assumptions the closed and the dissipative part
of such a generator
%are \textit{unique}, i.e.~when the two 
can be separated unambiguously.
%a quite natural question when trying to understand which part 
Given some GKSL-generator $L$, a sufficient condition they found is that if the Lindblad operators are traceless, then there exists a \textit{unique} traceless $H$ such that~\eqref{eq:gksl} holds, cf.~also \cite{Davies80unique}\,\footnote{
Actually, if the $V_j$ are traceless, it is known that the superoperator $-i[H,\cdot]$ is in some sense orthogonal to $\bf\Gamma$; for the precise statement,
refer to Sec.~\ref{sec_orth} below.
}.
This result has to be read with caution: it does \textit{not} state that the $V_j$ can be chosen unique, as this would be impossible given the ``unitary degree of freedom'' in the choice of the Lindblad operators, cf.~\cite[Eq.~(3.72) \& (3.73)]{BreuPetr02}.
However, this ambiguity can be eliminated by looking at the completely positive map \textit{generated} by the Lindblad operators, that is, by re-writing $L=-i[H,\cdot]-\sum_{j}( \frac12(V_j^*V_j(\cdot)+(\cdot)V_j^*V_j)-V_j(\cdot)V_j^* )$
%equivalently
as 
\begin{equation}\label{eq:L_Phi}
L=-i[H,\cdot]+\Phi-\Big\{\frac{\Phi^*({\bf1})}{2},\cdot\Big\}
\end{equation}
with $\Phi:=\sum_jV_j(\cdot)V_j^*$, where $\Phi^*({\bf1})=\sum_jV_j^*V_j$ and $\{\cdot,\cdot\}$ is the usual anticommutator.
On this level, the uniqueness result of \cite{GKS76} can be shown to be equivalent
(cf.~Lemma~\ref{lemma_trace_findim} below)
to the following:
For every GKSL-generator $L$ there exists a unique traceless Hamiltonian $H$ and a 
unique completely positive map $\Phi$ with ${\rm tr}(\Phi)=0$ such that~\eqref{eq:L_Phi} holds.
%\marginpar{formulation WRONG! trace zero does NOT resolve degree of freedom in Lindblad/Kraus operators. go to map level! check formulation of original papers and maybe adjust intro}
%\cite{Diss-Indra}

Taking another perspective, this result states that for Hamiltonians $H_1,H_2$ and \textit{traceless} completely positive maps $\Phi_1,\Phi_2$ a decomposition
$$
\Big(iH_2-iH_1+\frac{\Phi_1^*({\bf1})}{2}-\frac{\Phi_2^*({\bf1})}{2}\Big)(\cdot)+(\cdot)\Big(iH_2-iH_1+\frac{\Phi_1^*({\bf1})}{2}-\frac{\Phi_2^*({\bf1})}{2}\Big)^*=\Phi_1-\Phi_2
$$
is only possible if
%$H_1=H_2$ and, more importantly, 
$\Phi_1=\Phi_2$.
This is quite surprising for the following reason: the left-hand side of this equation is a Hermitian-preserving map\footnote{
We say $\Phi\in\mathcal L(\mathbb C^{n\times n})$ is Hermitian-preserving if $\Phi(X^*)=\Phi(X)^*$ for all $X\in\mathbb C^{n\times n}$.
}---let us abbreviate it by $K(\cdot)+(\cdot)K^*$---meaning
it can always be written as the difference of two completely positive maps \cite[Thm.~2.25]{Watrous18}.
Therefore, the condition of $\Phi_1,\Phi_2$ being traceless has to somehow prohibit such a decomposition.

The aim of this paper is to close this gap in understanding: how can it be that an decomposition $K(\cdot)+(\cdot)K^*=\Phi_1-\Phi_2$ into completely positive maps---which is always achievable---becomes impossible once the problem is restricted to traceless maps?
After all, the original proof of this \cite[Lemma~2.3]{GKS76} is a dimension-counting argument,
which does not provide any insight into this question.
%Part of GKS-result \cite[Lemma~2.3]{GKS76}: any $L\in\mathsf{GKSL}(n):=\mathsf L(\mathsf{CPTP}(n))$ decomposes into unique traceless Hamiltonian and unique Kossakowski matrix (w.r.t.~an arbitrary but fixed \textbf{traceless} ONB of $(\mathbb C^{n\times n},\langle\cdot,\cdot\rangle_{\sf HS}$)
%
%goals of this paper: generalize to $\mathsf L(\mathsf{CP}(n))=\{K(\cdot)+(\cdot)K^*+\Phi:K\in\mathbb C^{n\times n},\Phi\in\mathsf{CP}(n)\}$ 
%\& understand ``why'' this result holds (``problem'': the proof in \cite{GKS76} is not ``structural'' but rather ``just'' counts dimensions).
%After all, $K(\cdot)+(\cdot)K^*\in\mathsf{HP}(n)=\mathsf{CP}(n)-\mathsf{CP}(n)$ \cite[Thm.~2.25]{Watrous18} so why does restricting from $\mathsf{CP}(n)$ to $\mathsf{CP}_{\bf1}(n)$ make this impossible? (answer in Prop.~\ref{prop1} below)....
%and is there something special about ${\bf1}$/CP or does this hold for all $B$/all of $P$?
%
%$K(\cdot)+(\cdot)K^*$ (sometimes called \textit{generalized inner derivation}, cf.~\cite{FM92})
%
In fact, we will not only close this gap by presenting an alternative proof based
on a connection between weighted traces and expectation values of the Choi matrix, this proof even applies to the setting of completely positive dynamical semigroups as well as to more general conditions than the plain trace of $\Phi$.
More precisely, our main result reads as follows:\medskip\smallskip
%we will (1) strengthen the result (2) generalize it to arbitrary CP dynamics (3) generalize it to arbitrary ``weight matrices'' $B$ (i.e.~${\rm tr}(\cdot)=0$ becomes ${\rm tr}(B\,\cdot\,)=0$) under the sole condition that ${\rm tr}(B)\neq 0$. 

\noindent\textbf{Theorem~\ref{thm1} (Informal).}
{\it
Let any generator $L$ of a completely positive dynamical semigroup be given, that is, $e^{tL}$ is completely positive for all $t\geq 0$. Then for all $B$ with ${\rm Re}({\rm tr}(B))\neq 0$
there exist unique $K$ and a unique completely positive map $\Phi$ such that
\begin{itemize}
\item[(i)] the map $\Phi(B^*(\cdot)B)$ has trace zero (equivalently: for every set of Kraus operators $\{V_j\}_j$ of $\Phi$, it holds that ${\rm tr}(B^*V_j)=0$ for all $j$),
\item[(ii)] ${\rm Im}({\rm tr}(B^*K))=0$, and
\item[(iii)] $L=K(\cdot)+(\cdot)K^*+\Phi$.
\end{itemize}
}\smallskip

The condition ${\rm tr}(B^*V_j)=0$ for Lindblad operators 
has appeared in the past in the infinite-dimensional setting \cite[Thm.~30.16]{Parthasarathy92}; however,
the connection to unique decompositions of generators is new,
and this could be one way to establish an analogous result in infinite dimensions, cf.~Sec.~\ref{sec_outlook}

Now, Theorem~\ref{thm1} includes not only dynamics of open quantum systems (i.e.~completely positive trace-preserving evolutions), but also, for example,
probabilistic evolutions governed by operations (i.e.~completely positive trace-nonincreasing maps) \cite[Ch.~4]{Heinosaari12}.
%result for dynamics beyond CPTP (e.g., CPTNI = evolution of instruments (?)) dynamics.
%there, the generator being TA reduces to TNI which is an inequality constraint $\to$ dimension counting doesn't work anymore
In the case of trace-preserving dynamics, Theorem~\ref{thm1} boils down to the following:\medskip\smallskip

\noindent\textbf{Corollary~\ref{coro1} (Informal).}
{\it Given any GKSL-generator $L$ and any $B$ with ${\rm Re}({\rm tr}(B))\neq 0$,
there exists a unique Hermitian matrix $H$ as well as a unique completely positive map $\Phi$ such that
\begin{itemize}
\item[(i)] the map $\Phi(B^*(\cdot)B)$ has trace zero (equivalently: for every set of Kraus operators $\{V_j\}_j$ of $\Phi$ it holds that ${\rm tr}(B^*V_j)=0$ for all $j$),
\item[(ii)] ${\rm Re}({\rm tr}(B^* H))=\frac12{\rm Im}({\rm tr}(\Phi(B)))$, and
\item[(iii)] $L=-i[H,\cdot]+\Phi-\{\tfrac{\Phi^*({\bf1})}{2},\cdot\}$.
\end{itemize}
Moreover, if $B$ is Hermitian, then condition (ii) becomes ${\rm tr}(BH)=0$.} \medskip\smallskip

\noindent Setting $B={\bf1}$ in this corollary clearly recovers the uniqueness result of Gorini, Kossakowski, and Sudarshan.\medskip

This paper is organized as follows. In Sec.~\ref{sec:main} we introduce the mathematical tools we need, such as Lie wedges of semigroups of operators, or inner products of superoperators. Most importantly, this section features a lemma which connects general traces of completely positive maps with the trace of its Kraus operators (Lemma~\ref{lemma_trace_findim}).
Then Sec.~\ref{sec_mainres} is devoted to our main results: the precise formulations of Theorem~\ref{thm1} and Corollary~\ref{coro1} can be found in  Sec.~\ref{sec_mainres_unique},
while in Sec.~\ref{sec_orth} we show that the unique decomposition from Thm.~\ref{thm1} is even orthogonal with respect to an inner product induced by $B$.
Then in Sec.~\ref{sec_dec_pos} we ask to what extent complete positivity is necessary for Thm.~\ref{thm1} to hold: we will see that for general positive maps, such decompositions need not always exist.
Finally, we give some closing remarks and ideas for future research in Sec.~\ref{sec_outlook}

\section{Preliminaries}\label{sec:main}
%\marginpar{$\Phi^*$ dual, same as HS-adjoint iff $\Phi$ is HP}
First some notation:
given any vector space $\mathcal V$,
the collection of all linear maps $:\mathcal V\to\mathcal V$ will be denoted by $\mathcal L(\mathcal V)$.
If $\mathcal V=\mathbb C^{n\times n}$, then $\mathsf P(n)$ ($\mathsf{PTP}(n)$) denotes the subset of all positive (and trace-preserving) maps;
similarly, $\mathsf{CP}(n)$ ($\mathsf{CPTP}(n)$) is the collection of all completely positive (and trace-preserving) maps.
To test complete positivity, one usually employs the \textit{Choi matrix}: as proven in the seminal paper of Choi \cite{Choi75} a map $\Phi\in\mathcal L(\mathbb C^{n\times n})$ is completely positive if and only if 
$\mathsf C(\Phi):=({\rm id}\otimes\Phi)(|\Gamma\rangle\langle\Gamma|)$ is positive semi-definite, where $|\Gamma\rangle:=\sum_{j=1}^n|j\rangle\otimes|j\rangle$ is the (unnormalized) maximally entangled state.
%A well-known, yet important property of $|\Gamma\rangle$ is that it is independent of the chosen basis: for \textit{any} orthonormal basis $\{g_j\}_{j=1}^n$ of $\mathbb C^n$ one computes
%\begin{align*}
%\begin{split}
%|\Gamma\rangle&=\sum_{{j'}=1}^n|{j'}\rangle\otimes\Big(\sum_{j=1}^n\langle g_j|{j'}\rangle |g_j\rangle\Big)=\sum_{j=1}^n\Big(\sum_{{j'}=1}^n\langle g_j|{j'}\rangle|{j'}\rangle\Big)\otimes|g_j\rangle\,.
%\end{split}
%\end{align*}
%If we now define $\overline{x}$ as the complex conjugate of some vector $x\in\mathbb C^n$---that is,
%$\overline{x}:=
%\sum_{j=1}^n\overline{\langle j|x\rangle} |j\rangle=
%\sum_{j=1}^n\langle x|j\rangle |j\rangle$---the above calculation simplifies to
%$|\Gamma\rangle=\sum_{j=1}^n|\overline{g_j}\rangle\otimes|g_j\rangle$.
%Therefore, for every orthonormal basis $\{g_j\}_{j=1}^n$ of $\mathbb C^n$ the Choi matrix satisfies
%\begin{equation}\label{eq:Choi_matrix_rep}
%\mathsf C(\Phi)=\sum_{j,k=1}^n|\overline{g_j}\rangle\langle \overline{g_k}|\otimes\Phi(|g_j\rangle\langle g_k|)\,.
%\end{equation}

We also need some objects from Lie theory: $\mathsf U(n)$ (resp.~$\mathsf{SU}(n)$) denotes the unitary (resp.~special unitary) group in $n$ dimensions and $\mathfrak u(n)$ (resp.~$\mathfrak{su}(n)$) is the corresponding Lie algebra.
In other words, $i\mathfrak u(n)$ ($i\mathfrak{su}(n)$) is the collection of all Hermitian (and traceless) $n\times n$-matrices.
Moreover, given some real or complex finite-dimensional vector space $\mathcal V$ as well as a closed subsemigroup with identity\footnote{
This means that $S$ is a closed subset of $\mathcal L(\mathcal V)$ such that ${\rm id}\in S$
and that for all $\Phi_1,\Phi_2\in S$ it holds that $\Phi_1\circ\Phi_2\in S$.
}
$S\subseteq\mathcal L(\mathcal V)$ one defines its \textit{Lie wedge}
$\mathsf L(S)$ \cite{HHL89} as the collection of all generators of dynamical semigroups in $S$, that is,
\begin{equation*}%\label{eq:liewedge}
\mathsf L(S):=\{A\in\mathcal L(\mathcal V): e^{tA}\in S\text{ for all }t\geq 0\}\,.
\end{equation*}
In this language, the celebrated result of Gorini, Kossakowski, Sudarshan \cite{GKS76}, and Lindblad \cite{Lindblad76} reads
\begin{equation}\label{eq:liewedge_cptp}
\mathsf L(\mathsf{CPTP}(n))=\big\{-i[H,\cdot]+\Phi-\{\tfrac{\Phi^*({\bf1})}{2},\cdot\}\,:\,H\in i\mathfrak u(n),\Phi\in\mathsf{CP}(n)\big\}\,,
\end{equation}
cf.~also \cite{DHKS08}; here $\Phi^*$ is the unique dual of $\Phi$, i.e.~${\rm tr}(B\Phi(A))={\rm tr}(\Phi^*(B)A)$ for all $A,B\in\mathbb C^{n\times n}$.
%The elements of $\mathsf L(\mathsf{CPTP}(n))$ shall be called GKSL\textit{-generators}.
Moreover, 
the Lie wedge of the completely positive maps is known
to be given by
\cite[Thm.~3]{Lindblad76}
\begin{equation}\label{eq:liewedge_cp}
\mathsf L(\mathsf{CP}(n))=\big\{K(\cdot)+(\cdot)K^*+\Phi\,:\,K\in\mathbb C^{n\times n},\Phi\in\mathsf{CP}(n)\big\}\,.
\end{equation}

Now---like in the original GKS-paper \cite{GKS76}---we need to equip $\mathcal L(\mathbb C^{n\times n})$ with an inner product.
For this we use that every real or complex finite-dimensional vector space
$(\mathcal V,\langle\cdot,\cdot\rangle_{\mathcal V})$ induces an inner product on $\mathcal L(\mathcal V)$ via
\begin{align}\label{eq:innerprod_LV}
\begin{split}
\langle\cdot,\cdot\rangle: \mathcal L(\mathcal V)\times\mathcal L(\mathcal V)&\to\mathbb C\\
(\Phi,\Psi)&\mapsto \sum_\alpha\langle \Phi(v_\alpha),\Psi(v_\alpha)\rangle_{\mathcal V}\,,
\end{split}
\end{align}
where $\{v_\alpha\}_\alpha$ is an arbitrary orthonormal basis of $(\mathcal V,\langle\cdot,\cdot\rangle_{\mathcal V})$; unsurprisingly, 
the value of $\langle\Phi,\Psi\rangle$
is independent of the chosen orthonormal basis.
For example, choosing $\mathcal V=\mathbb C^n$ equipped with the standard inner product gives rise to the \textit{Hilbert-Schmidt} or \textit{Frobenius
inner product} $\langle\cdot,\cdot\rangle_{\sf HS}$ on $\mathbb C^{n\times n}$, that is, $\langle A,B\rangle_{\sf HS}={\rm tr}(A^*B)$ for all $A,B\in \mathbb C^{n\times n}$ \cite[Ch.~IV.2]{Bhatia97}.
Going one level higher, choosing $(\mathcal V,\langle\cdot,\cdot\rangle_{\mathcal V})=(\mathbb C^{n\times n},\langle\cdot,\cdot\rangle_{\sf HS})$ yields an inner product on $\mathcal L(\mathbb C^{n\times n})$ which acts like 
$\langle\Phi,\Psi\rangle=\sum_{j=1}^{n^2}\langle \Phi(G_j),\Psi(G_j)\rangle_{\sf HS}$ for any orthonormal basis
$\{G_j\}_{j=1}^{n^2}$ of $(\mathbb C^{n\times n},\langle\cdot,\cdot\rangle_{\sf HS})$.
In particular, this recovers the trace of linear maps on square matrices\footnote{
Recall that, given any real or complex finite-dimensional vector space $\mathcal V$, the trace is the unique \textit{commutative} linear functional on the linear maps $:\mathcal V\to\mathcal V$ such that 
${\rm tr}({\bf1}_{\mathcal V})=\dim{\mathcal V}$.
%The key step in the proof is $${\rm tr}(|g_j\rangle\langle g_k|)={\rm tr}(|g_j\rangle\langle g_j|\circ|g_j\rangle\langle g_k|)={\rm tr}(|g_j\rangle\langle g_k|\circ |g_j\rangle\langle g_j|)={\rm tr}(0)=0$$ and, analogously, ${\rm tr}(|g_j\rangle\langle g_j|)={\rm tr}(|g_k\rangle\langle g_k|)$ for all $j,k$ where $\{g_j\}_{j=1}^{\dim{\mathcal V}}$ is an orthonormal basis of $\mathcal V$ w.r.t.~an arbitrary inner product.
Therefore, evaluating ${\rm tr}(\Phi)$ via Eq.~\eqref{eq:trace_Phi} is unambiguous
as it is independent of the chosen orthonormal basis.
}
via
\begin{align}\label{eq:trace_Phi}
\begin{split}
{\rm tr}(\Phi)=\langle{\rm id},\Phi\rangle&=\sum_{j,k=1}^{n}\langle \,|g_j\rangle\langle g_k|\,,\Phi(|g_j\rangle\langle g_k|)\rangle_{\sf HS}\\
&=\sum_{j,k=1}^{n}{\rm tr}\big(  |g_k\rangle\langle g_j|\Phi(|g_j\rangle\langle g_k|) \big)=\sum_{j,k=1}^{n}\langle g_j|\Phi(|g_j\rangle\langle g_k|) |g_k\rangle
\end{split}
\end{align}
for all $\Phi\in\mathcal L(\mathbb C^{n\times n})$, where
$\{g_j\}_{j=1}^n$ is an arbitrary orthonormal basis of $\mathbb C^n$ (as then $\{|g_j\rangle\langle g_k|\}_{j,k=1}^n$ is an orthonormal basis of $(\mathbb C^{n\times n},\langle\cdot,\cdot\rangle_{\sf HS})$).

Below we will frequently deal with maps $\Phi_{A,B}\in\mathcal L(\mathbb C^{n\times n})$ of product form, i.e.~$\Phi_{A,B}(X):=AXB$ for some $A,B\in\mathbb C^{n\times n}$.
These maps admit some special properties, as is readily verified:
\begin{itemize}
\item The trace of $\Phi_{A,B}$ equals
\begin{equation}\label{eq:trace_AB}
{\rm tr}(\Phi_{A,B})\overset{\eqref{eq:trace_Phi}}=\sum_{j,k=1}^{n}\langle g_j|A|g_j\rangle\langle g_k|B |g_k\rangle={\rm tr}(A){\rm tr}(B)\,.
\end{equation}
\item The $\langle\cdot,\cdot\rangle_{\sf HS}$--adjoint\footnote{
Given any $\Phi\in\mathcal L(\mathbb C^{n\times n})$, $\Phi^\dagger$ is the unique linear map such that $\langle \Phi^\dagger(X),Y\rangle_{\sf HS}=\langle X,\Phi(Y)\rangle_{\sf HS}$ for all $X,Y\in\mathbb C^{n\times n}$.
Note that, one has $\Phi^\dagger=\Phi^*$ if (and only if) $\Phi$ is Hermitian-preserving. 
The adjoint can be used to derive a coordinate-free representation of the inner product:
$\langle\Phi,\Psi\rangle={\rm tr}(\Phi^\dagger\circ\Psi)$
\label{footnote_HS_adjoint}
}
of $\Phi_{A,B}$ is given by $\Phi_{A,B}^\dagger(X)=A^*XB^*$.
\end{itemize}

Now this formalism lets us relate properties of Kraus operators (resp.~their trace) to properties of the corresponding completely positive map:

\begin{lemma}\label{lemma_trace_findim}
Given any $B\in\mathbb C^{n\times n}$ and any $\Phi\in\mathsf{CP}(n)$ the following statements are equivalent.
\begin{itemize}
\item[(i)] ${\rm tr}(\Phi\circ\Phi_{B^*,B})=0$. In abuse of notation we will---here and henceforth---write ${\rm tr}(\Phi\circ\Phi_{B^*,B})={\rm tr}(\Phi(B^*(\cdot)B))$.
\item[(ii)] $\langle B(\cdot)B^*,\Phi\rangle=0$
\item[(iii)] For all $\{V_j\}_{j=1}^m$, $m\in\mathbb N$ which satisfy $\Phi\equiv\sum_{j=1}^mV_j(\cdot)V_j^*$ it holds that ${\rm tr}(B^* V_j)=0$ for all $j=1,\ldots,m$.
\item[(iv)] There exist $\{V_j\}_{j=1}^m$, $m\in\mathbb N$ such that $\Phi\equiv\sum_{j=1}^mV_j(\cdot)V_j^*$ as well as ${\rm tr}(B^* V_j)=0$ for all $j=1,\ldots,m$.
\end{itemize}
\end{lemma}
\begin{proof}
``(i) $\Leftrightarrow$ (ii)'':
Using footnote~\ref{footnote_HS_adjoint} as well as cyclicity of the trace we compute
 $\langle B(\cdot)B^*,\Phi\rangle={\rm tr}((B(\cdot)B^*)^\dagger\circ\Phi)={\rm tr}((B^*(\cdot)B)\circ\Phi)={\rm tr}(\Phi(B^*(\cdot)B))$.

``(iii) $\Rightarrow$ (iv)'': Every completely positive map admits Kraus operators \cite[Thm.~1]{Choi75}.
 ``(iv) $\Rightarrow$ (i) $\Rightarrow$ (iii)'': Follows at once from the fact that for all $\Phi\in\mathsf{CP}(n)$ and any set of Kraus operators $\{V_j\}_{j=1}^m\subset\mathbb C^{n\times n}$, $m\in\mathbb N$ of $\Phi$ (i.e.~$\Phi=\sum_{j=1}^m V_j(\cdot)V_j^*$) it holds that
\begin{align*}
{\rm tr}(\Phi(B^*(\cdot)B))=\sum_{j=1}^m{\rm tr}(V_jB^*(\cdot)BV_j^*)\overset{\eqref{eq:trace_AB}}=\sum_{j=1}^m|{\rm tr}(B^*V_j)|^2\,.\tag*{\qedhere}
\end{align*}
%as then ${\rm tr}(B V_j)=0$ for all $j=1,\ldots,m$ if and only if ${\rm tr}(\Phi(B(\cdot)B^*))=0$.
\end{proof}
\noindent In other words ${\rm tr}(\Phi(B^*(\cdot)B)=0$ is equivalent to all Kraus operators of $\Phi$ being orthogonal to $B$ (w.r.t.~$\langle\cdot,\cdot\rangle_{\sf HS}$).\medskip

Let us draw some connections of this lemma to different notions from the literature:
\begin{remark}\label{rem_entfid}
${}$ %this placeholder is to ensure that label/ref works, without this \ref{rem_entfid} breaks and links to the beginning of the document!!
\begin{itemize}
\item[(i)] Given $\Phi\in\mathsf{CPTP}(n)$ and a quantum state $\rho$, Schumacher \cite{Schumacher96} first introduced the entanglement fidelity of $\Phi$ w.r.t.~$\rho$, denoted\footnote{
While not important to our work, $F_e(\Phi,\rho)$ is defined as the fidelity of $({\rm id}\otimes\Phi)(|\psi_\rho\rangle\langle\psi_\rho|)$ and $|\psi_\rho\rangle\langle\psi_\rho|$ (i.e.~$F_e(\Phi,\rho):=\langle\psi_\rho|({\rm id}\otimes\Phi)(|\psi_\rho\rangle\langle\psi_\rho|)|\psi_\rho\rangle$) where $\psi_\rho$ is any purification of $\rho$.
}
$F_e(\Phi,\rho)$.
For any Kraus representation $\sum_jV_j(\cdot)V_j^*$ of $\Phi$ this quantity is equal to
$\sum_j|{\rm tr}(V_j\rho)|^2$, cf.~\cite[Lemma~10.10]{Holevo12}.
Thus the proof of Lemma~\ref{lemma_trace_findim} shows $F_e(\Phi,\rho)={\rm tr}(\Phi(\rho(\cdot)\rho))$, even for general $\Phi\in\mathsf{CP}(n)$.
\item[(ii)] 
One of the equivalent conditions for uniquely decomposing generators of dynamical semigroups given by Davies \cite[Thm.~1]{Davies80unique} is that $\int_{\mathsf U(n)\times\mathsf U(n)}U^*\Phi(UV)V^*\,dU\,dV=0$.
He proved this by observing that this integral evaluates to $(\sum_{j=1}^m|{\rm tr}(V_j)|^2)\frac{\bf1}{n}$.
Thus, by the proof of Lemma~\ref{lemma_trace_findim} $\int_{\mathsf U(n)\times\mathsf U(n)}U^*\Phi(UV)V^*\,dU\,dV={\rm tr}(\Phi)\frac{\bf1}{n}$ for all $\Phi\in\mathsf{CP}(n)$; in
other words $\Phi\mapsto \int_{\mathsf U(n)\times\mathsf U(n)}U^*\Phi(UV)V^*\,dU\,dV$ is a channel analogue of twirling \cite[Ex.~2.24]{Heinosaari12}.
\end{itemize}
\end{remark} 

The above inner product $\langle\cdot,\cdot\rangle$ on $\mathcal L(\mathbb C^{n\times n})$ is of course not unique in any way.
% for an inner product on $\mathcal L(\mathbb C^{n\times n})$. 
As the matrix $B$ from Lemma~\ref{lemma_trace_findim} will play the role of a reference state,
we will need an inner product (with certain properties) which is somehow induced by this $B$.
There are many ways to accomplish this, but the construction we follow starts from a $B$-inner product on $\mathbb C^{n\times n}$ which can be turned into a $B$-inner product on the linear maps on $\mathbb C^{n\times n}$.
% via Eq.~\eqref{eq:innerprod_LV}.
A common choice for the former is $(X,Y)\mapsto {\rm tr}(X^*\sqrt{\omega }Y\sqrt \omega )$ from quantum thermodynamics,
where $\omega \in\mathbb C^{n\times n}$ is an arbitrary positive definite matrix which plays the role of the Gibbs state of a system.
Thus, given any $B\in\mathbb C^{n\times n}$ positive definite, the inner product $\langle X,Y\rangle_B:={\rm tr}(X^*BYB)$
%into Eq.~\eqref{eq:innerprod_LV}
gives rise to (in abuse of notation)
\begin{equation}\label{eq:innerprod_L_B}
\langle \Phi,\Psi\rangle_B:=\sum_\alpha\langle \Phi(G_\alpha),\Psi(G_\alpha)\rangle_B=\sum_\alpha{\rm tr}\big( \Phi(G_\alpha)^*B\Psi(G_\alpha)B\big)
\end{equation}
for all $\Phi,\Psi\in \mathcal L(\mathbb C^{n\times n})$
and any orthonormal basis $\{G_\alpha\}_\alpha$ of $(\mathbb C^{n\times n},\langle\cdot,\cdot\rangle_{\sf HS})$.
A straightforward computation shows that $\langle \cdot,\cdot\rangle_B$ is related to $\langle \cdot,\cdot\rangle$ via
\begin{align*}
\langle \Phi,\Psi\rangle_B= {\rm tr}(\Phi^\dagger\circ ( B(\cdot)B )\circ\Psi) &=\langle \sqrt B\Phi(\cdot)\sqrt B,\sqrt B\Psi(\cdot)\sqrt B\rangle\\
&= \langle \Phi_{\sqrt B,\sqrt B} \circ\Phi, \Phi_{\sqrt B,\sqrt B} \circ\Psi\rangle \,.
\end{align*}
As a consequence, Lemma~\ref{lemma_trace_findim} (ii) is equivalent to $\langle{\rm id},\Phi\rangle_B=0$.

Now for multiplication maps this $B$-inner product evaluates to
%\begin{align}
%\langle X_1(\cdot)Y_1,X_2(\cdot)Y_2\rangle_B&=\langle \sqrt B X_1(\cdot)Y_1\sqrt B,\sqrt B X_2(\cdot)Y_2\sqrt B\rangle\notag\\
%&={\rm tr}\big((\sqrt B X_1(\cdot)Y_1\sqrt B)^*\circ(\sqrt B X_2(\cdot)Y_2\sqrt B)\big) \notag\\
%&={\rm tr}\big((\sqrt B X_1)^*(\cdot)(Y_1\sqrt B)^*)\circ(\sqrt B X_2(\cdot)Y_2\sqrt B)\big)  \notag\\
%&={\rm tr}\big((X_1^* B X_2)(\cdot)(\sqrt B Y_1^*Y_2\sqrt B)\big)  \notag\\
%&\overset{\eqref{eq:trace_Phi}}={\rm tr}(X_1^* B X_2){\rm tr}(Y_2BY_1^*)\,.\label{eq:innerprod_XYB}
%\end{align}
\begin{align}
\langle X_1(\cdot)Y_1,X_2(\cdot)Y_2\rangle_B&
=\sum_\alpha{\rm tr}( Y_1^*G_\alpha^*X_1^*BX_2G_\alpha Y_2B )\notag\\
&
=\sum_\alpha{\rm tr}( G_\alpha^*X_1^*BX_2G_\alpha Y_2BY_1^* )\notag\\
&={\rm tr}\big((X_1^* B X_2)(\cdot)(Y_2BY_1^*)\big)  \overset{\eqref{eq:trace_Phi}}={\rm tr}(X_1^* B X_2){\rm tr}(Y_2BY_1^*)\,,\label{eq:innerprod_XYB}
\end{align}
where $X_1,X_2,Y_1,Y_2\in\mathbb C^{n\times n}$ were arbitrary.
\section{Main Results}\label{sec_mainres}

One of the central objects of the GKS-paper \cite{GKS76}---completely positive maps with traceless Kraus (resp.~Lindblad) operators---by Lemma~\ref{lemma_trace_findim}
are precisely those $\Phi\in\mathsf{CP}(n)$ of zero trace.
Because this concept---as well as its generalization to arbitrary ``weight matrices'' $B\in\mathbb C^{n\times n}$---is central to this work let us turn it into a proper definition:

\begin{defi}\label{def_1}
Given any $B\in\mathbb C^{n\times n}$ define
$
\mathsf{CP}_B(n)$ as the collection of all $\Phi\in \mathsf{CP}(n)$ which satisfy ${\rm tr}(\Phi(B^*(\cdot)B))=0$.
\end{defi}
\noindent Thus, by Lemma~\ref{lemma_trace_findim}, $\mathsf{CP}_{\bf1}(n)$ is the set of all completely positive maps with traceless Kraus operators.\medskip

At this point we emphasize that
(i) $\mathsf{CP}_B(n)$ is non-empty for all $B$ (as a simple consequence of Lemma~\ref{lemma_trace_findim}~(iv)),
and that (ii)
not every completely positive map is in $\mathsf{CP}_B(n)$ for some (non-zero) $B\in\mathbb C^{n\times n}$, i.e.
$$
\bigcup_{\{B\in\mathbb C^{n\times n}:B\neq 0\}}\mathsf{CP}_B(n)\subsetneq\mathsf{CP}(n)\,.
$$
The simplest example which shows that this inclusion is strict is the following:
\begin{example}\label{ex_depol}
Consider the completely depolarizing channel $D(X):={\rm tr}(X)\frac{\bf1}n$
which has Kraus operators $\{n^{-1/2}|j\rangle\langle k|\}_{j,k=1}^n$, cf.~\cite[Ch.~2.2.3]{Watrous18}.
Now assume that $D\in \mathsf{CP}_B(n)$ for some $B\in\mathbb C^{n\times n}$, $B\neq 0$.
Then Lemma~\ref{lemma_trace_findim} shows $\langle j|B|k\rangle=0$
for all $j,k=1,\ldots,n$,
i.e.~$B=0$, 
a contradiction.
\end{example}

\subsection{Unique Decompositions of Generators of Completely Positive Dynamical Semigroups}\label{sec_mainres_unique}

In order to further explore the sets $
\mathsf{CP}_B(n)$ and their properties we need to invoke the concept of \textit{vectorization}, cf.~\cite[Ch.~4.2 ff.]{HJ2}.
Given any $B\in\mathbb C^{m\times n}$ one defines 
%its vectorization via 
$\vec B:=\sum_{j=1}^n\sum_{k=1}^m B_{kj}|j\otimes k\rangle\in\mathbb C^n\otimes\mathbb C^m\simeq\mathbb C^{mn}$ (where $B_{kj}:=\langle k|B|j\rangle$).
Thus, $\vec B$ is the vector resulting from stacking all columns of $B$ underneath each other, which can be seen by means of the identity $\vec B=\sum_{j=1}^n|j\rangle\otimes B|j\rangle=({\bf 1}\otimes B)|\Gamma\rangle$.
%The key property of this formalism for our purpose is the following: for all $A\in\mathbb C^{l\times m},B\in\mathbb C^{m\times n},C\in\mathbb C^{n\times p}$, $l,m,n,p\in\mathbb N$
%it holds that
%\begin{equation}\label{eq:vec_prod}
%\vec{ABC}=(C^\top\otimes A)\vec B
%\end{equation}
%where ${}^\top$ denotes the usual transposition.\medskip
%%Eq.~\eqref{eq:vec_prod} yields alternative proof of $|\Gamma\rangle=\vec{\bf1}$
%
While this insight allows for a more elegant proof of the following key lemma---which relates the trace of maps of the form $\Phi(X(\cdot)Y)$ to the Choi matrix of $\Phi$---we prefer to present an elementary, yet not much longer proof.

\begin{lemma}\label{lemma0}
For all $\Phi\in\mathcal L(\mathbb C^{n\times n})$, $n\in\mathbb N$ and all $X,Y\in\mathbb C^{n\times n}$
\begin{equation*}
{\rm tr}\big(\Phi(X(\cdot)Y)\big)=\langle\vec{X^*}|\mathsf C(\Phi)|\vec Y\rangle\,.
\end{equation*}
%where $\mathsf C(\Phi):=({\rm id}\otimes\Phi)(|\Gamma\rangle\langle\Gamma|)$ is the Choi matrix of $\Phi$ and $|\Gamma\rangle:=\sum_{j=1}^n|j\otimes j\rangle$ is the maximally entangled state.
\end{lemma}
\begin{proof}
We compute
\begin{align*}
\langle\vec{X^*}|&\mathsf C(\Phi)|\vec Y\rangle\\
&=\sum_{a,b,j,k,p,q=1}^m\big\langle  (X^*)_{ba}\,a\otimes b\,\big| \,|q\rangle\langle k|\otimes\Phi(|q\rangle\langle k|)\,   \big| \, Y_{jp}\,p\otimes j  \big\rangle\\
&=\sum_{a,b,j,k,p,q=1}^m(X_{ab}^*)^*Y_{jp}\langle a|q\rangle\langle k|p\rangle\langle  b|\Phi(|q\rangle\langle k|)|j\rangle\\
&=\sum_{a,b,j,k=1}^n X_{ab}Y_{jk}\langle b|\Phi(|a\rangle\langle k|)|j\rangle\,.
\end{align*}
Now note that
\begin{align*}
{\rm tr}\big( \Phi\big( |a\rangle\langle b|(\cdot)|j\rangle\langle k|   \big) \big)&\overset{\eqref{eq:trace_Phi}}=\sum_{p,q=1}^n \big\langle p\big| \Phi\big( |a\rangle\langle b|p\rangle\langle q|j\rangle\langle k|   \big)\big|q\big\rangle\\
&=\sum_{p,q=1}^n\delta_{bp}\delta_{qj}\langle p|\Phi(|a\rangle\langle k|)|q\rangle=\langle b|\Phi(|a\rangle\langle k|)|j\rangle
\end{align*}
so combining these two identities yields
\begin{align*}
\langle\vec{X^*}|\mathsf C(\Phi)|\vec Y\rangle&=\sum_{a,b,j,k=1}^n X_{ab}Y_{jk}{\rm tr}\big( \Phi\big( |a\rangle\langle b|(\cdot)|j\rangle\langle k|   \big) \big)\\
&={\rm tr}\Big(\Phi\Big(  \sum_{a,b=1}^nX_{ab}|a\rangle\langle b|(\cdot)\sum_{j,k=1}^nY_{jk}|j\rangle\langle k|  \Big)\Big)
\end{align*}
which is equal to ${\rm tr}(\Phi(X(\cdot)Y))$, as claimed.
\end{proof}

The importance of this identity is due to the following insight: If $\Phi$ is completely positive and ${\rm tr}(\Phi(B^*(\cdot)B))=0$, then the Choi matrix of $\Phi$ has $\vec B$ is in the kernel:

\begin{lemma}\label{lemma1}
For all $B\in\mathbb C^{n\times n}$ and all $\Phi\in\mathsf{CP}_{B}(n)$ one has $\mathsf C(\Phi)|\vec B\rangle=0$.
\end{lemma}
\begin{proof}
First, $\Phi\in\mathsf{CP}(n)$ implies $\mathsf C(\Phi)\geq 0$.
Together with ${\rm tr}\,\Phi(B^*(\cdot)B)=0$---which by Lemma~\ref{lemma0}
is equivalent to $\langle\vec B|\mathsf C(\Phi)|\vec B\rangle=0$---this shows
\begin{align*}
0&=\langle\vec B|\mathsf C(\Phi)|\vec B\rangle\\
&=\big\langle\sqrt{\mathsf C(\Phi)}\vec B\big|\sqrt{\mathsf C(\Phi)}\vec B\big\rangle=\big\|\sqrt{\mathsf C(\Phi)}\vec B\big\|^2\,.
\end{align*}
Therefore $0=\sqrt{\mathsf C(\Phi)}0=\sqrt{\mathsf C(\Phi)}(\sqrt{\mathsf C(\Phi)}\vec B)=\mathsf C(\Phi)\vec B$ which concludes the proof.
\end{proof}

\noindent This lemma also explains ``why'' Example~\ref{ex_depol} works as it does:
Any element of $\mathsf{CP}_{B}(n)$ necessarily has Choi rank (equivalently: Kraus rank, cf.~\cite[Ch.~2.2.2]{Watrous18}) strictly less than $n^2$---regardless of the chosen $B$---but the Choi matrix of the completely depolarizing channel is a multiple of the identity, hence its Kraus rank is $n^2$.\medskip

Now, the kernel property of Lemma~\ref{lemma1} is the reason why maps of the form $K(\cdot)+(\cdot)K^*$---sometimes called \textit{generalized inner ${}^*$-derivation} \cite{FM92}---are ``incompatible'' with $\mathsf{CP}_{B}(n)$ in the following sense:
\begin{proposition}\label{prop1}
Let $B,K\in\mathbb C^{n\times n}$ with ${\rm tr}(B)\neq 0$ be given. The following statements are equivalent.
\begin{itemize}
\item[(i)] $K(\cdot)+(\cdot)K^*\in\mathsf{CP}_{B}(n)-\mathsf{CP}_{B}(n)$
\item[(ii)] There exists $\lambda\in\mathbb R$ such that $K=i\lambda{\bf1}$.
\end{itemize}
\end{proposition}
\begin{proof}
``(ii) $\Rightarrow$ (i)'': Assume $K=i\lambda{\bf1}$ for some $\lambda\in\mathbb R$. Then one computes $K(\cdot)+(\cdot)K^*=0=0-0\in \mathsf{CP}_{B}(n)-\mathsf{CP}_{B}(n)\,.$
``(i) $\Rightarrow$ (ii)'': Assume there exist $\Phi_1,\Phi_2\in\mathsf{CP}_B(n)$ such that $K(\cdot)+(\cdot)K^*=\Phi_1-\Phi_2$.
By Lemma~\ref{lemma1} $\mathsf C(\Phi_1)|\vec B\rangle=\mathsf C(\Phi_2)|\vec B\rangle=0$, hence $\mathsf C(K(\cdot)+(\cdot)K^*)|\vec B\rangle=0$, as well.
%Now, as in the proof of Lemma~\ref{lemma0}, there exist $\lambda\in\mathbb R^n$ and an orthonormal basis $\{g_j\}_{j=1}^n$ of $\mathbb C^n$ such that $B=\sum_{j=1}^n\lambda_j|g_j\rangle\langle g_j|$;
%with this we show that $\vec B\in{\sf ker}(\mathsf C(K(\cdot)+(\cdot)K^*))$ implies that $K$ is diagonal (w.r.t.~the eigenbasis of $B$).
Then
%using Eqs.~\eqref{eq:Choi_matrix_rep} \& \eqref{eq:vecB} 
one for all $j,k=1,\ldots,n$ finds
\begin{align*}
0&=\langle k\otimes j|\mathsf C(K(\cdot)+(\cdot)K^*)| \vec B\rangle\\
&=\sum_{a,b,p,q=1}^nB_{qp}\big\langle k\otimes j\big|\,|{  a}\rangle\langle {  b}|\otimes \big(K| a\rangle\langle b|+| a\rangle\langle b|K^*\big)\,\big|{  p}\otimes q\big\rangle\\
&=\sum_{a,b,p,q=1}^nB_{qp}\langle k|a\rangle\langle b|p\rangle \langle 
j|K| a\rangle\langle b| q \rangle+\sum_{a,b,p,q=1}^nB_{qp}\langle k|a\rangle\langle b|p\rangle \langle j | a\rangle\langle b|K^*|q \rangle\\
&=\sum_{p=1}^nB_{pp}\langle 
j|K| k\rangle+\delta_{jk}\sum_{p,q=1}^nB_{qp}\langle p|K^*|q \rangle={\rm tr}(B)\langle j|K| k\rangle+\delta_{jk}{\rm tr}(B K^*)\,.
\end{align*}
If $j\neq k$, this shows ${\rm tr}(B)\langle j|K|k\rangle=0$; but ${\rm tr}(B)\neq 0$ by assumption, hence $\langle j|K|k\rangle=0$ 
and thus all off-diagonal elements of $K$ vanish.
On the other hand, setting $k=j$ shows
$
\langle j|K| j\rangle=-\frac{{\rm tr}(B K^*)}{{\rm tr}(B)}
$
for all $j$ so, altogether, $K=c{\bf 1}$ for some $c\in\mathbb C$.
For the final step note that
%$K(\cdot)+(\cdot)K^*\in\mathsf{CP}_{B}(n)-\mathsf{CP}_{B}(n)$ in particular means that
\begin{align*}
0=0-0&\overset{\hphantom{\eqref{eq:trace_AB}}}= {\rm tr}(\Phi_1(B^*(\cdot)B))-{\rm tr}(\Phi_2(B^*(\cdot)B)) \\
&\overset{\hphantom{\eqref{eq:trace_AB}}}={\rm tr}\big( KB^*(\cdot)B+B^*(\cdot)B K^* \big)\\
&\overset{\eqref{eq:trace_AB}}={\rm tr}(KB^*){\rm tr}(B)+{\rm tr}(B^*){\rm tr}(B K^*)=2\,{\rm Re}\big({\rm tr}(B){\rm tr}(KB^*)\big)\,.
\end{align*}
Inserting $K=c{\bf 1}$ then shows $0=2\,{\rm Re}(c)|{\rm tr}(B)|^2$ so---because ${\rm tr}(B)\neq 0$ by assumption---there exists $\lambda\in\mathbb R$ such that $c=i\lambda$; hence (ii) holds.
\end{proof}

\noindent This in some sense generalizes the known fact that no (non-trivial) map $K(\cdot)+(\cdot)K^*$ is completely positive: 
Indeed, this proposition shows that the only map $K(\cdot)+(\cdot)K^*$ in the \textit{vector space generated by} $\mathsf{CP}_B(n)$ is the zero map.
As noted in the introduction, this is quite non-trivial, because every $K(\cdot)+(\cdot)K^*$ is Hermitian-preserving, and the set of Hermitian-preserving maps 
is precisely the vector space generated by $\mathsf{CP}(n)$.
From this point-of-view, Proposition~\ref{prop1} shows that (and why) the trace condition ${\rm tr}(\Phi(B^*(\cdot)B))=0$ is the key to
separating $\mathsf{CP}(n)-\mathsf{CP}(n)$ from $\{K(\cdot)+(\cdot)K^*:K\in\mathbb C^{n\times n}\}$.\medskip

With this proposition at hand we are finally ready to state and prove our main result:

\begin{thm}\label{thm1} 
Given any $B\in\mathbb C^{n\times n}$ with ${\rm Re}({\rm tr}(B))\neq 0$, the map
\begin{align*}
\Xi_B:\{K\in\mathbb C^{n\times n}:{\rm Im}({\rm tr}(B^* K))=0\}\times\mathsf{CP}_B(n)&\to \mathsf L(\mathsf{CP}(n))\\
(K,\Phi)&\mapsto K(\cdot)+(\cdot)K^*+\Phi
\end{align*}
is bijective, where $\mathsf L(\mathsf{CP}(n))$ is the Lie wedge of
$\mathsf{CP}(n)$, cf.~Eq.~\eqref{eq:liewedge_cp}.
Equivalently, the map
\begin{equation}\label{eq:thm1_thetaB1}
\Theta_B:(H,Z,\Phi)\mapsto -i[H,\cdot]+\{Z,\cdot\}+\Phi\in\mathsf L(\mathsf{CP}(n))
\end{equation}
with domain $\{(H,Z)\in i\mathfrak u(n)\times i\mathfrak u(n):
\substack{{\rm Im}({\rm tr}(B^* Z))={\rm Re}({\rm tr}(B^* H))}
\}\times \mathsf{CP}_B(n)$
is bijective.
In particular, if $B\in\mathbb C^{n\times n}$ is Hermitian with ${\rm tr}(B)\neq 0$, then
\begin{align*}
\Theta_B:\{H\in i\mathfrak u(n):{\rm tr}(B H)=0\}\times i\mathfrak u(n)\times \mathsf{CP}_B(n)&\to \mathsf L(\mathsf{CP}(n))\\
(H,Z,\Phi)&\mapsto -i[H,\cdot]+\{Z,\cdot\}+\Phi
\end{align*}
is bijective.
\end{thm}
\begin{proof}
It suffices to prove the statement regarding $\Xi_B$, as then bijectivity of $\Theta_B$ from Eq.~\eqref{eq:thm1_thetaB1} follows at once by decomposing $K=Z-iH$ with $H,Z\in i\mathfrak u(n)$.

First we prove that $\Xi_B$ is surjective. Given 
any generator $L\in \mathsf L(\mathsf{CP}(n))$,
by~\eqref{eq:liewedge_cp} there exist $K_0\in\mathbb C^{n\times n}$, $\Phi_0\in\mathsf{CP}(n)$ such that $L=K_0(\cdot)+(\cdot)K_0^*+\Phi_0$.
On the other hand there exist $\{V_j\}_{j=1}^m$, $m\leq n^2$ such that $\Phi_0=\sum_{j=1}^mV_j(\cdot)V_j^*$ \cite[Rem.~6]{Choi75}.
Defining $\tilde V_j:=V_j -\frac{{\rm tr}(B^*V_j)}{{\rm tr}(B^*)}{\bf 1} $
and
$$
\tilde K:=K_0+\sum_{j=1}^m\Big(\frac{{\rm tr}(B^*V_j)}{{\rm tr}(B^*)}\Big)^*V_j-\frac{\sum_{j=1}^m|{\rm tr}(B^*V_j)|^2}{2|{\rm tr}(B)|^2}{\bf1}\,,
$$
one finds $L=\tilde K(\cdot)+(\cdot)\tilde K^*+\sum_{j=1}^m\tilde V_j(\cdot)\tilde V_j^*$, cf.~\cite[Eq.~(1.4)]{Davies80unique}.
%{\color{red}this lets us compute}
%\begin{align*}
%L
%%&=K_0(\cdot)+(\cdot)K_0^*+\sum_{j=1}^mV_j(\cdot)V_j^*\\
%&=K_0(\cdot)+(\cdot)K_0^*+\sum_{j=1}^m\big(\tilde V_j +{\rm tr}(BV_j){\bf 1} \big)(\cdot)\big(\tilde V_j +{\rm tr}(BV_j){\bf 1} \big)^*\\
%&=\Big(\underbrace{K_0+\sum_{j=1}^m{\rm tr}(BV_j)^*\tilde V_j+\frac{\sum_{j=1}^m|{\rm tr}(BV_j)|^2}{2}{\bf1}}_{=:\tilde K}\Big)(\cdot)\\
%&\qquad+(\cdot)\Big(K_0+\sum_{j=1}^m{\rm tr}(BV_j)^*\tilde V_j+\frac{\sum_{j=1}^m|{\rm tr}(BV_j)|^2}{2}{\bf1}\Big)^*+\sum_{j=1}^m\tilde V_j(\cdot)\tilde V_j^* \,.
%\end{align*}
By definition, ${\rm tr}(B^*\tilde V_j)=0$ for all $j=1,\ldots,m$ so  Lemma~\ref{lemma_trace_findim} shows that $\Phi:=\sum_{j=1}^m\tilde V_j(\cdot)\tilde V_j^*$ is in $\mathsf{CP}_B(n)$.
Thus all that is left is to ``shift'' $\tilde K$ such that ${\rm Im}({\rm tr}(B^* \tilde K))=0$.
Obviously, replacing $\tilde K$ by $\tilde K+i\lambda{\bf 1}$, $\lambda\in\mathbb R$ does not change $L$.
For the shifted matrix the trace condition then reads
$$
0={\rm Im}({\rm tr}(B^* (\tilde K+i\lambda{\bf 1})))={\rm Im}({\rm tr}(B^* \tilde K))+\lambda{\rm Re}({\rm tr}(B^*))\,,
$$
so setting $K:=\tilde K-i\frac{{\rm Im}({\rm tr}(B^* \tilde K))}{{\rm Re}({\rm tr}(B))}{\bf 1}$ shows $(K,\Phi)\in{\rm dom}(\Xi_B)$ as well as $\Xi_B(K,\Phi)=
\tilde K(\cdot)+(\cdot)\tilde K^*+\Phi
=L$, as desired.

For injectivity, assume
$K_1(\cdot)+(\cdot)K_1^*+\Phi_1=K_2(\cdot)+(\cdot)K_2^*+\Phi_2$
for some $K_1,K_2\in\mathbb C^{n\times n}$ , $\Phi_1,\Phi_2\in\mathsf{CP}_B(n)$ such that ${\rm Im}({\rm tr}(B^* K_j))=0$, $j=1,2$; equivalently,
$$
(K_2-K_1)(\cdot)+(\cdot)(K_2-K_1)^*=\Phi_1-\Phi_2\in \mathsf{CP}_B(n)-\mathsf{CP}_B(n)\,.
$$
Thus by Proposition~\ref{prop1} $K_1=K_2+i\lambda{\bf1}$ for some $\lambda\in\mathbb R$.
This has two consequences: on the one hand, 
this imaginary difference between $K_1$ and $K_2$ cancels in the sense that
$K_1(\cdot)+(\cdot)K_1^*=K_2(\cdot)+(\cdot)K_2^*$;
this in turn implies $\Phi_1=\Phi_2$. 
On the other hand, $K_1=K_2+i\lambda{\bf1}$ together with the trace condition on $K_1,K_2$ yields
\begin{align*}
0={\rm Im}({\rm tr}(B^* K_1))
%={\rm Im}({\rm tr}(B^* (   K_2+i\lambda{\bf1}  )))
={\rm Im}({\rm tr}(B^* K_2))+\lambda{\rm Re}({\rm tr}(B^*))=\lambda{\rm Re}({\rm tr}(B))\,.
\end{align*}
But ${\rm Re}({\rm tr}(B))\neq 0$ by assumption so $\lambda$ has to vanish, meaning $K_1=K_2$.
This concludes the proof.
\end{proof}

As a special case we obtain a family of unique decompositions for arbitrary generators of quantum-dynamical semigroups;
all one has to do is combine Theorem~\ref{thm1} with the trace-annihilation condition (i.e.~$Z=-\frac12\Phi^*({\bf1})$ in Eq.~\eqref{eq:thm1_thetaB1}) to obtain:

\begin{corollary}\label{coro1}
For all $B\in\mathbb C^{n\times n}$ with ${\rm Re}({\rm tr}(B))\neq 0$ the map
\begin{align*}
\hat\Xi_B:(H,\Phi)&\mapsto -i[H,\cdot]+\Phi-\Big\{\frac{\Phi^*({\bf1})}{2},\cdot\Big\}\in\mathsf L(\mathsf{CPTP}(n))
\end{align*}
with domain $\{(H,\Phi)\in i\mathfrak u(n)\times \mathsf{CP}_B(n):\substack{{\rm Im}({\rm tr}(\Phi(B)))=2\,{\rm Re}({\rm tr}(B^* H))}\}
$
is bijective,
where $\mathsf L(\mathsf{CPTP}(n))$ is the Lie wedge of $\mathsf{CPTP}(n)$, cf.~Eq.~\eqref{eq:liewedge_cptp}.
In particular, if $B\in\mathbb C^{n\times n}$ is Hermitian with ${\rm tr}(B)\neq 0$, then
\begin{align*}
\hat\Xi_B:\{H\in i\mathfrak u(n):{\rm tr}(BH)=0\}\times \mathsf{CP}_B(n)&\to \mathsf L(\mathsf{CPTP}(n))\\
(H,\Phi)&\mapsto -i[H,\cdot]+\Phi-\Big\{\frac{\Phi^*({\bf1})}{2},\cdot\Big\}
\end{align*}
is bijective.
\end{corollary}
\noindent The only non-trivial step regarding how to adjust the domain of $\hat\Xi_B$ is that
$
{\rm Im}({\rm tr}(B^*\Phi^*({\bf1})))={\rm Im}({\rm tr}(\Phi(B^*)))=-{\rm Im}({\rm tr}(\Phi(B)))
$,
where in the second step one uses that complete positivity implies that $\Phi$ is Hermitian-preserving.\medskip

Setting $B={\bf1}$ in Coro.~\ref{coro1} yields that
% \cite[Lemma~2.3]{GKS76} that
\begin{align*}
\begin{split}\hat\Xi_{\bf1}:i\mathfrak{su}(n)\times\mathsf{CP}_{\bf1}(n)&\to\mathsf L(\mathsf{CPTP}(n))\\(H,\Phi)&\mapsto -i[H,\cdot]+\Phi-\Big\{\frac{\Phi^*({\bf1})}{2},\cdot\Big\}\end{split}
\end{align*}
is bijective
which recovers the old uniqueness result of Gorini, Kossakowski, and Sudarshan,
as discussed in the introduction.
%In other words: every GKSL-generator decomposes uniquely into a closed system part with a traceless Hamiltonian, and into a dissipative part generated by a \textit{traceless} completely positive map.
%
%\hrulefill
%
%.......
%
%if $B=|\psi\rangle\langle\psi|$, then ${\rm tr}(\Phi(B(\cdot)B^*))=\langle\psi|\Phi(|\psi\rangle\langle\psi|)|\psi\rangle=0$; in other words $\Phi(|\psi\rangle\langle\psi|)\in\{|\psi\rangle\langle\psi|\}^\perp$.
%If $\Phi$ is positive this---as in the proof of Lemma~\ref{lemma1}---implies that $\psi\in\ker{\Phi(|\psi\rangle\langle\psi|)}$

This result begs the question: what is the effect of different choices of the matrix $B$ on the unique decomposition of generators?
Let us shine some light on this by means of a simple example:

\begin{example}\label{ex_bloch}
Consider
% a two-level atom coupled to the vacuum, undergoing spontaneous emission'' \cite[Ch.~8.4]{NC00} (special case of  
the Lindblad operators which generate the Bloch equations, i.e.~$H=\frac{\omega}{2}\sigma_z$, $V_1=\sqrt{\gamma_1}|0\rangle\langle 1|$, $V_2=\sqrt{\gamma_2}|1\rangle\langle 0|$, $V_3=\sqrt{\gamma_3}\sigma_z$ for some $\omega\in\mathbb R$, $\gamma_1,\gamma_2,\gamma_3\geq 0$, cf.~\cite[Ch.~8.4]{NC00}, \cite[Ch.~5.5]{AM11}.
This gives rise to the generator $L=-i[H,\cdot]-\sum_{j=1}^3( \frac12(V_j^*V_j(\cdot)+(\cdot)V_j^*V_j)-V_j(\cdot)V_j^* )$.
Invoking the superoperator matrix---which, given some $\Psi\in\mathcal L(\mathbb C^{n\times n})$ is the unique matrix $\hat\Psi\in\mathbb C^{n^2\times n^2}$ such that $\vec{\Psi(X)}=\hat\Psi\vec X$ for all $X\in\mathbb C^{n\times n}$---one finds the familiar form
\begin{align*}
\hat L=\begin{pmatrix}
-\gamma_2&0&0&\gamma_1\\
0&-\frac{\gamma_1+\gamma_2}{2}-2\gamma_3+i\omega&0&0\\
0&0&-\frac{\gamma_1+\gamma_2}{2}-2\gamma_3-i\omega&0\\
\gamma_2&0&0&-\gamma_1
\end{pmatrix}\,.
\end{align*}
Now let any $B\in\mathbb C^{2\times 2}$ with ${\rm Re}({\rm tr}(B))\neq 0$ be given and consider $\hat\Xi_B^{-1}(L)=(H_B,\Phi_B)$ from Corollary~\ref{coro1}.
Using the short-hand notation
$\beta_{jk}:=\frac{b_{jk}}{b_{11}+b_{22}}$
and
%\quad\text{ and }\quad 
$\beta:=\frac{\gamma_1\beta_{12}-\gamma_2\beta_{21}^*}2$
one computes
\begin{align*}
\widehat{-i[H_B,\cdot]}=\begin{pmatrix}
0&\beta&\beta^*&0\\
-\beta^*&i\omega+2\gamma_3(\beta_{11}^*-\beta_{11})&0&\beta^*\\
-\beta&0&-i\omega+2\gamma_3(\beta_{22}^*-\beta_{22})&\beta\\
0&-\beta&-\beta^*&0
\end{pmatrix}
\end{align*}
and
\begin{equation}\label{eq:ex_gammaB}
\widehat{{\bf\Gamma}_B}=\begin{pmatrix}
\gamma_2&\beta&\beta^*&-\gamma_1\\
-\beta^*& \frac{\gamma_1+\gamma_2}{2}+ 2\gamma_3( \beta_{11}^*+\beta_{22} ) &0&\beta^*\\
-\beta&0&\frac{\gamma_1+\gamma_2}{2}+ 2\gamma_3( \beta_{11}+\beta_{22}^* ) &\beta\\
-\gamma_2&-\beta&-\beta^*&\gamma_1
\end{pmatrix}\,,
\end{equation}
%\begin{align*}
%\hat{\bf\Gamma}_B=\begin{pmatrix}
% -\gamma_{2} & \frac{1}{2} (\frac{\gamma_{2} b_{21}^*}{b_{11}^*+b_{22}^*}-\frac{b_{12} \gamma_{1}}{b_{11}+b_{22}}) & \frac{1}{2} (\frac{b_{21} \gamma_{2}}{b_{11}+b_{22}}-\frac{\gamma_{1} b_{12}^*}{b_{11}^*+b_{22}^*}) & \gamma_{1} \\
% \frac{1}{2} (\frac{\gamma_{1} b_{12}^*}{b_{11}^*+b_{22}^*}-\frac{b_{21} \gamma_{2}}{b_{11}+b_{22}}) & -\frac{2 \gamma_{3} b_{11}^*}{b_{11}^*+b_{22}^*}-\frac{\gamma_{1}+\gamma_{2}}{2}-\frac{2 b_{22} \gamma_{3}}{ b_{11}+b_{22}} & 0 & \frac{1}{2} (\frac{b_{21} \gamma_{2}}{b_{11}+b_{22}}-\frac{\gamma_{1} b_{12}^*}{b_{11}^*+b_{22}^*}) \\
% \frac{1}{2} (\frac{b_{12} \gamma_{1}}{b_{11}+b_{22}}-\frac{\gamma_{2} b_{21}^*}{b_{11}^*+b_{22}^*}) & 0 & \frac{2 \gamma_{3} b_{11}^*}{b_{11}^*+b_{22}^*}-\frac{(b_{11}+b_{22}) (\gamma_{1}+\gamma_{2})+4 \gamma_{3} (2 b_{11}+b_{22})}{2 (b_{11}+b_{22})} & \frac{1}{2} (\frac{\gamma_{2} b_{21}^*}{b_{11}^*+b_{22}^*}-\frac{b_{12} \gamma_{1}}{b_{11}+b_{22}}) \\
% \gamma_{2} & \frac{1}{2} (\frac{b_{12} \gamma_{1}}{b_{11}+b_{22}}-\frac{\gamma_{2} b_{21}^*}{b_{11}^*+b_{22}^*}) & \frac{1}{2} (\frac{\gamma_{1} b_{12}^*}{b_{11}^*+b_{22}^*}-\frac{b_{21} \gamma_{2}}{b_{11}+b_{22}}) & -\gamma_{1}
%\end{pmatrix}\,,
%\end{align*}
where ${\bf\Gamma}_B:=\{\tfrac{\Phi_B^*({\bf1})}{2},\cdot\}-\Phi_B$ so $L=-i[H_B,\cdot]-{\bf\Gamma}_B$.
Therefore, if $B$ is such that $\beta=0$ and $b_{11}^*b_{22}\in\mathbb R$ (the latter being equivalent to $\beta_{11}^*=\beta_{11}$), one recovers the original 
%components $H_B=H$ and ${\bf\Gamma}_B={\bf\Gamma}$.
off-diagonal scaling of $[H,\cdot],{\bf\Gamma}$.
Moreover, if $B\in i\mathfrak u(n)$, then~\eqref{eq:ex_gammaB} simplifies to
\begin{align*}
\widehat{{\bf\Gamma}_B}=\begin{pmatrix}
\gamma_2&\frac{b_{12}}2(\gamma_1-\gamma_2)&\frac{b_{12}^*}2(\gamma_1-\gamma_2)&-\gamma_1\\
-\frac{b_{12}^*}2(\gamma_1-\gamma_2)& \frac{\gamma_1+\gamma_2}{2}+ 2\gamma_3&0&\frac{b_{12}^*}2(\gamma_1-\gamma_2)\\
-\frac{b_{12}}2(\gamma_1-\gamma_2)&0&\frac{\gamma_1+\gamma_2}{2}+ 2\gamma_3&\frac{b_{12}}2(\gamma_1-\gamma_2)\\
-\gamma_2&-\frac{b_{12}}2(\gamma_1-\gamma_2)&-\frac{b_{12}^*}2(\gamma_1-\gamma_2)&\gamma_1
\end{pmatrix}\,,
\end{align*}
\end{example}
In summary, for the dynamical semigroup governing the Bloch equations,
a general Hermitian $B$ (via Coro.~\ref{coro1}) yields a dissipator which need not be covariant anymore.
Moreover, waiving Hermiticity of $B$ leads to imaginary $\gamma_3$-contributions in both the dissipative and, equally, the Hamiltonian part.

%Thus, in this example, different $B$'s only affect how $[H,\cdot],\Gamma$ mixes diagonal and off-diagonal elements
%
%Then
%$$
%B=\begin{pmatrix}
%a&\gamma_1c\\\gamma_2c^*&b
%\end{pmatrix}
%$$
%for all $a,b\in\mathbb R$, $a+b\neq 0$ and all $c\in\mathbb C$ leads to the same decomposition:
%$$
%\widehat{-i[H,\cdot]}=\begin{pmatrix}
%0&0&0&0\\0&i\omega&0&0\\0&0&-i\omega&0\\0&0&0&0
%\end{pmatrix}$$
%and
%$$
%\hat\Gamma=\begin{pmatrix}
%-\gamma_2&0&0&\gamma_1\\0&-2\gamma_3-\frac12(\gamma_1+\gamma_2)&0&0\\0&0&-2\gamma_3-\frac12(\gamma_1+\gamma_2)&0\\
%\gamma_2&0&0&-\gamma_1
%\end{pmatrix}
%$$
%(observation: independent of $\gamma_3$)
%In particular this holds for all diagonal $B$.

\subsection{Orthogonality of the Generator's Unique Components}\label{sec_orth}

The zero-trace condition does not only guarantee a unique decomposition of $L\in\mathsf L(\mathsf{CPTP}(n))$ into closed and dissipative components, it also ensures that these two are orthogonal to each other.
Indeed, the Lindblad operators being traceless is a known sufficient condition for mutual orthogonality of the components of $L$ \cite[Lemma~2.4]{Diss-Indra}: A straightforward computation---together with Eq.~\eqref{eq:trace_AB}---shows
\begin{align*}
\langle i[H,\cdot],{\bf\Gamma}\rangle&=-i\,{\rm tr}([H,\cdot]^\dagger\circ{\bf\Gamma})\\
&=-i\,{\rm tr}([H,\cdot]\circ{\bf\Gamma})
%&
=-i\,{\rm tr}({\bf\Gamma}([H,\cdot]))=\sum_j 2\,{\rm Im}\big( {\rm tr}(V_j){\rm tr}(V_j^*H) \big)
\end{align*}
for any $H\in i\mathfrak u(n)$ and any choice of Lindblad operators $\{V_j\}_j$ of ${\bf\Gamma}$, meaning that if ${\rm tr}(V_j)=0$ for all $j$ (equivalently: ${\rm tr}(\Phi)=0$, Lemma~\ref{lemma_trace_findim}), then $\langle i[H,\cdot],{\bf\Gamma}\rangle=0$.
This orthogonality of $i[H,\cdot]$ and ${\bf\Gamma}$ in the traceless case turns out to be special, and it does not hold for general $B$ (see Prop.~\ref{prop_orth} below).
To see this we first need the following lemma:

\begin{lemma}\label{lemma_innerproducts}
Given any $H,Z,B\in i\mathfrak u(n)$ with $B$ positive definite, as well as any $\Phi\in\mathsf{CP}(n)$ the following statements hold.
\begin{itemize}
\item[(i)] $\langle i[H,\cdot],\{Z,\cdot\}\rangle_B=i\,{\rm tr}(B){\rm tr}(Z[B,H])$
\item[(ii)] $\langle i[H,\cdot],\Phi\rangle_B=\sum_j 2\,{\rm Im}({\rm tr}(HBV_j){\rm tr}(BV_j)^*)$ where $\{V_j\}_j$ is an arbitrary choice of Kraus operators of $\Phi$.
\item[(iii)] $\langle \{Z,\cdot\},\Phi\rangle_B=\sum_j 2\,{\rm Re}({\rm tr}(ZBV_j){\rm tr}(BV_j)^*)$ where $\{V_j\}_j$ is an arbitrary choice of Kraus operators of $\Phi$.
\end{itemize}
In particular, if $\Phi\in\mathsf{CP}_B(n)$, then $\langle i[H,\cdot],\Phi\rangle_B=\langle \{Z,\cdot\},\Phi\rangle_B=0$.
\end{lemma}
\begin{proof}
These are straightforward computations using linearity \& Eq.~\eqref{eq:innerprod_XYB}. (i):
\begin{align*}
\langle i[H,\cdot],\{Z,\cdot\}\rangle_B&= -i\langle H(\cdot),Z(\cdot)\rangle_B-i\langle H(\cdot),(\cdot)Z\rangle_B\\
&\qquad\qquad+i\langle (\cdot)H,Z(\cdot)\rangle_B+i\langle (\cdot)H,(\cdot)Z\rangle_B\\
&=-i\,{\rm tr}(H B Z){\rm tr}(B)-i\,{\rm tr}(HB ){\rm tr}(ZB)\\
&\qquad\qquad+i\,{\rm tr}( B Z){\rm tr}(BH)+i\,{\rm tr}( B ){\rm tr}(ZBH)\\
&=i\,{\rm tr}(B){\rm tr}( ZBH-ZHB)=i\,{\rm tr}(B){\rm tr}(Z[B,H])\,.
\end{align*}
(ii): 
\begin{align*}
\langle i[H,\cdot],\Phi\rangle_B&=-i\sum_j\langle H(\cdot),V_j(\cdot)V_j^*\rangle_B
+i\sum_j\langle(\cdot)H,V_j(\cdot)V_j^*\rangle_B\\
&=-i\sum_j{\rm tr}(H B V_j){\rm tr}(V_j^*B)
+i\sum_j{\rm tr}( B V_j){\rm tr}(V_j^*B H)\\
&=-i\sum_j\big({\rm tr}(H B V_j){\rm tr}(BV_j)^*-{\rm tr}(H B V_j)^*{\rm tr}(BV_j)  \big)  \\
&= \sum_j 2\,{\rm Im}({\rm tr}(HBV_j){\rm tr}(BV_j)^*) \,.
\end{align*}
The computation for (iii) is analogous.
Finally, the additional statement follows from (ii) \& (iii) together with the fact that $\Phi\in\mathsf{CP}_B(n)$ is equivalent to $0={\rm tr}(B^* V_j)={\rm tr}(B V_j)$, cf.~Lemma~\ref{lemma_trace_findim}.
\end{proof}

This lemma lets us quickly verify which (resp.~under what conditions/in what sense) parts of dynamical generators are orthogonal.

\begin{proposition}\label{prop_orth}
Given any $B\in i\mathfrak u(n)$ positive definite, the following statements hold.
\begin{itemize}
\item[(i)] The set $\{K(\cdot)+(\cdot)K^*:K\in\mathbb C^{n\times n}\}$ is orthogonal to $\mathsf{CP}_B(n)$ with respect to the modified inner product $\langle\cdot,\cdot\rangle_B$ from Eq.~\eqref{eq:innerprod_L_B}.
\item[(ii)] The set $\{-i[H,\cdot]:H\in i\mathfrak u(n), {\rm tr}(HB)=0\}$ is
orthogonal to the set $\{\Phi-\{\tfrac{\Phi^*({\bf1})}{2},\cdot\}\,:\,\Phi\in\mathsf{CP}_B(n)\}$ with respect to $\langle\cdot,\cdot\rangle_B$ if and only if $B$ is a multiple of the identity.
\end{itemize}
\end{proposition}
\begin{proof}
(i): Decomposing $K=Z-iH$ with $Z,H\in i\mathfrak u(n)$ one sees that $\langle K(\cdot)+(\cdot)K^*,\Phi\rangle_B=\langle -i[H,\cdot],\Phi\rangle_B+\langle \{Z,\cdot\},\Phi\rangle_B$ which by Lemma~\ref{lemma_innerproducts} is zero.

(ii): Before proving this statement let us do the following general computation:
using Lemma~\ref{lemma_innerproducts},
one for all $H\in i\mathfrak u(n)$ and all $\Phi\in\mathsf{CP}(n)$ finds
\begin{align}\label{eq:prop_orth_1}
\begin{split}
\big\langle i[H,\cdot],\Phi-\big\{&\tfrac{\Phi^*({\bf1})}{2},\cdot\big\}\big\rangle_B=\langle i[H,\cdot],\Phi\rangle_B-\big\langle i[H,\cdot],
\big\{\tfrac{\Phi^*({\bf1})}{2},\cdot\big\}\big\rangle_B\\
&=\sum_j 2\,{\rm Im}({\rm tr}(HBV_j){\rm tr}(BV_j)^*)-\tfrac{i}2{\rm tr}(B){\rm tr}( \Phi^*({\bf1})[B,H] )\,.
\end{split}
\end{align}

``$\Leftarrow$'': W.l.o.g.~$B={\bf1}$. Then Eq.~\eqref{eq:prop_orth_1} implies 
$\langle i[H,\cdot],\Phi-\{\tfrac{\Phi^*({\bf1})}{2},\cdot\}\rangle=\sum_j 2\,{\rm Im}({\rm tr}(HV_j){\rm tr}(V_j)^*$.
But all summands vanish because ${\rm tr}(V_j)=0$ for all $j$ by Lemma~\ref{lemma_trace_findim}.
% because $\Phi\in\mathsf{CP}_{\bf1}(n)$

``$\Rightarrow$'': We prove this by contrapositive. W.l.o.g.~$n>1$.
Assume that $B$ is not a multiple of the identity, i.e.~$B=\sum_j b_j|g_j\rangle\langle g_j|$ where $\{g_j\}_j$ is some orthonormal basis of $\mathbb C^n$ and the $b_j$ are positive numbers such that $b_j\neq b_k$ for some $j\neq k$.
Define $H:=b_k|g_j\rangle\langle g_j|+|g_j\rangle\langle g_k|+|g_k\rangle\langle g_j|-b_j|g_k\rangle\langle g_k|$ and $V_0:=i|g_j\rangle\langle g_j|+|g_j\rangle\langle g_k|-|g_k\rangle\langle g_k|$.
Moreover, based on this, define $V:=V_0-\frac{{\rm tr}(BV_0)}{{\rm tr}(B)}{\bf 1}$ and $\Phi:=V(\cdot)V^*$.
The following statements are readily verified:
\begin{itemize}
\item $H\in i\mathfrak u(n)$ with ${\rm tr}(BH)=0$
\item $[B,H]=(b_j-b_k)(|g_j\rangle\langle g_k|-|g_k\rangle\langle g_j|)$
\item ${\rm tr}(BV)=0$, which by Lemma~\ref{lemma_trace_findim} shows $\Phi\in\mathsf{CP}_B(n)$
\item ${\rm tr}(BV_0)=ib_j-b_k$
\end{itemize}
We claim that
\begin{equation}\label{eq:prop_orth_2}
\big\langle i[H,\cdot],\Phi-\big\{\tfrac{\Phi^*({\bf1})}{2},\cdot\big\}\big\rangle_B
%=-\tfrac{i}2{\rm tr}(B){\rm tr}(\Phi([B,H]))
=(b_j-b_k)({\rm tr}(B)-b_j)\,,
\end{equation}
which
%, in particular, is not zero because $B>0$ (so ${\rm tr}(B)>b_j$) and $b_j\neq b_k$. This 
would conclude the proof because then we found elements from either set that are not $\langle\cdot,\cdot\rangle_B$-orthogonal (as $b_j\neq b_k$, and ${\rm tr}(B)>b_j$ due to $B>0$).
Indeed, ${\rm tr}(BV)=0$ together with Eq.~\eqref{eq:prop_orth_1} implies
\begin{align*}
\big\langle i[H,\cdot],\Phi-\big\{\tfrac{\Phi^*({\bf1})}{2},\cdot\big\}\big\rangle_B&=-\tfrac{i}2\,{\rm tr}(B){\rm tr}(V^*V[B,H])\\
&=-\tfrac{i}2(b_j-b_k){\rm tr}(B)\big(\langle Vg_k|Vg_j\rangle-\langle Vg_j|Vg_k\rangle\big)\,.
\end{align*}
As $V|g_j\rangle=(i-\frac{ib_j-b_k}{{\rm tr}(B)})|g_j\rangle$ and $V|g_k\rangle=|g_j\rangle-(1+\frac{ib_j-b_k}{{\rm tr}(B)})|g_k\rangle$, the overlap of these vectors amounts to $\langle Vg_k|Vg_j\rangle=({\rm tr}(B))^{-1}
(b_k+i({\rm tr}(B)-b_j))$.
With this we arrive at
\begin{align*}
\big\langle i[H,\cdot],\Phi-\big\{\tfrac{\Phi^*({\bf1})}{2},\cdot\big\}\big\rangle_B&=(b_j-b_k){\rm tr}(B){\rm Im}(\langle Vg_k|Vg_j\rangle)\\
&=(b_j-b_k){\rm tr}(B)({\rm tr}(B))^{-1}({\rm tr}(B)-b_j)\,,
\end{align*}
i.e.~Eq.~\eqref{eq:prop_orth_2} holds and we are done.
\end{proof}
Rephrasing Proposition~\ref{prop_orth} (i), the decomposition from Theorem~\ref{thm1}
is not only unique, but even orthogonal with respect to the modified inner product $\langle\cdot,\cdot\rangle_B$.

\section{Positivity Instead of Complete Positivity}\label{sec_dec_pos}

As Theorem~\ref{thm1} generalizes unique decompositions of
$\mathsf{CPTP}$-dynamics to arbitrary $\mathsf{CP}$-dynamics, a natural
question is whether the assumption of complete positivity is necessary, or whether this result continues to hold in more general scenarios (e.g., positive dynamics)?
A naive first guess would be that complete positivity is essential as positivity of the Choi matrix was central to our proof.
To substantiate this, in the following we will investigate decompositions of the transposition map ${(\cdot)}^\top$, which is the prime example of a positive map
which is not completely positive.

It turns out that ${(\cdot)}^\top$---which is in $\mathsf L(\mathsf P(n))$, i.e.~it generates a positive dynamical semigroup, cf.~\cite[Ex.~6]{Chru14}---can be decomposed into $K(\cdot)+(\cdot)K^*+\Phi$ with $K\in\mathbb C^{n\times n}$, ${\rm Im}({\rm tr}(B^* K))=0$, and $\Phi\in\mathsf{P}_B(n)$ only for \textit{some}---but not
\textit{all}---$B\in\mathbb C^{n\times n}$ with ${\rm Re}({\rm tr}(B))\neq 0$.
More precisely, we find the following characterization of when such a decomposition exists:
\begin{proposition}\label{prop_transpos_decomp}
Given any $B\in\mathbb C^{n\times n}$ with ${\rm Re}({\rm tr}(B))\neq 0$, the following statements are equivalent.
\begin{itemize}
\item[(i)] There exist $K\in\mathbb C^{n\times n}$ and $\Phi\in\mathsf P(n)$ such that ${\rm Im}({\rm tr}(B^*K))=0$, ${\rm tr}(\Phi(B^*(\cdot)B)=0$,
and ${(\cdot)}^\top=K(\cdot)+(\cdot)K^*+\Phi$.
\item[(ii)] ${\rm tr}(B\overline{B})\leq 0$ where $\overline{B}:=(B^*)^\top$ is the entrywise complex conjugate.
\end{itemize}
\end{proposition}
\begin{proof}
``(ii) $\Rightarrow$ (i)'': Choose $K:=\frac{{\rm tr}(B\overline{B}){\rm tr}(B)}{2{\rm Re}({\rm tr}(B))|{\rm tr}(B)|^2}{\bf 1}$ and $\Phi:={(\cdot)}^\top-\frac{{\rm tr}(B\overline{B})}{|{\rm tr}(B)|^2}{\rm id}$, and note that these expressions
are well defined because ${\rm Re}({\rm tr}(B))\neq 0$ by assumption.
Then ${\rm tr}(B^*K)=\frac{{\rm tr}(B\overline{B})}{|{\rm tr}(B)|^2}$ so ${\rm Im}({\rm tr}(B^*K))=0$ due to ${\rm tr}(B\overline{B})\leq 0$ (which in particular means ${\rm tr}(B\overline{B})\in\mathbb R$).
Moreover, $-{\rm tr}(B\overline{B})\geq 0$ guarantees that $\Phi$ is
positive as a non-negative linear combination of positive maps.
The decomposition $K(\cdot)+(\cdot)K^*+\Phi={(\cdot)}^\top$ is readily verified and, finally, Eq.~\eqref{eq:trace_Phi} \& Eq.~\eqref{eq:trace_AB} yield
\begin{align}
{\rm tr}(\Phi(B^*(\cdot)B)&={\rm tr}(B^\top (\cdot)^\top\overline{B})-\frac{{\rm tr}(B\overline{B})}{|{\rm tr}(B)|^2}{\rm tr}(B^*(\cdot)B)\notag\\
&= \sum_{j,k}\langle j|B^\top|k\rangle\langle j|\overline{B}|k\rangle-\frac{{\rm tr}(B\overline{B})}{|{\rm tr}(B)|^2}|{\rm tr}(B)|^2\notag\\
&=  \sum_{j,k}\langle k|B|j\rangle\langle j|\overline{B}|k\rangle-
{\rm tr}(B\overline{B})=0\,,\label{eq:tr_BBline}
\end{align}
as desired.

``(i) $\Rightarrow$ (ii)'': 
Assume that there exist $K\in\mathbb C^{n\times n}$, $\Phi\in\mathsf P(n)$ such that ${\rm Im}({\rm tr}(B^*K))=0$, ${\rm tr}(\Phi(B^*(\cdot)B)=0$, and ${(\cdot)}^\top=K(\cdot)+(\cdot)K^*+\Phi$.
If we can show that ${(\cdot)}^\top-K(\cdot)-(\cdot)K^*(=\Phi)$ being positive implies $K=c\cdot{\bf 1}$ for some $c\in\mathbb C$ with ${\rm Re}(c)\leq 0$, then we would be done:
%Then ${\rm Im}(c)=0$ so $K=\lambda\cdot{\bf1}$ for some $\lambda\leq 0$. 
%On the other hand, 
by Eq.~\eqref{eq:trace_AB} \& Eq.~\eqref{eq:tr_BBline}
\begin{align*}
0={\rm tr}(\Phi(B^*(\cdot)B)&={\rm tr}\big( B^\top(\cdot)^\top\overline{B} \big)-{\rm tr}( KB^*(\cdot)B )-{\rm tr}( B^*(\cdot)BK^* )\\
&={\rm tr}(B\overline{B})-2{\rm Re}(c)|{\rm tr}(B)|^2
\end{align*}
i.e.~${\rm Re}(c)=\frac{{\rm tr}(B\overline{B})}{2|{\rm tr}(B)|^2}$.
But ${\rm Re}(c)\leq 0$ so this would imply ${\rm tr}(B\overline{B})\leq 0$, as desired.

Now assume that $K\in\mathbb C^{n\times n}$ is any matrix such that $\Phi:={(\cdot)}^\top-K(\cdot)-(\cdot)K^*$ is positive.
%, then $K=c\cdot{\bf 1}$ for some $c\in\mathbb C$ with ${\rm Re}(c)\leq 0$.
Given any $j\neq k$, a direct computation shows
$
\langle k|\Phi(|j\rangle\langle j|)|k\rangle=0
$. But $\Phi$ is positive so all principal minors of $\Phi(|j\rangle\langle j|)$ are non-negative \cite[Coro.~7.1.5]{HJ1}; in particular,
\begin{align*}
0&\leq \det\begin{pmatrix}
\langle j|\Phi(|j\rangle\langle j|)|j\rangle&\langle j|\Phi(|j\rangle\langle j|)|k\rangle\\
\langle k|\Phi(|j\rangle\langle j|)|j\rangle&\langle k|\Phi(|j\rangle\langle j|)|k\rangle
\end{pmatrix}\\
&= \det\begin{pmatrix}
\langle j|\Phi(|j\rangle\langle j|)|j\rangle&\langle j|\Phi(|j\rangle\langle j|)|k\rangle\\
\langle k|\Phi(|j\rangle\langle j|)|j\rangle&0
\end{pmatrix}=-\big|\langle j|\Phi(|j\rangle\langle j|)|k\rangle  \big|^2
\end{align*}
so
$0=\langle j|\Phi(|j\rangle\langle j|)|k\rangle=-K_{jk}
$
for all $j\neq k$, 
meaning $K$ has to be diagonal, i.e.~$K={\rm diag}(K_{11},\ldots,K_{nn})$ for some $K_{11},,\ldots,K_{nn}\in\mathbb C$.
Next note that, by assumption, $\Phi(|{\bf e}\rangle\langle {\bf e}|)\geq 0$ where ${\bf e}:=(1,\ldots,1)^\top$. Again all principal minors have to be non-negative so---because $K$ is diagonal---we for all $j\neq k$ compute
\begin{align*}
0&\leq\det\begin{pmatrix}
\langle j| \Phi(|{\bf e}\rangle\langle {\bf e}|) |j\rangle&\langle j| \Phi(|{\bf e}\rangle\langle {\bf e}|) |k\rangle\\\langle k| \Phi(|{\bf e}\rangle\langle {\bf e}|) |j\rangle&\langle k| \Phi(|{\bf e}\rangle\langle {\bf e}|) |k\rangle
\end{pmatrix}\\
&=\det\begin{pmatrix}
1-\langle j|K|{\bf e}\rangle-\langle{\bf e}|K^*| j\rangle&1-\langle j|K|{\bf e}\rangle-\langle{\bf e}|K^*|k\rangle\\
1-\langle k|K|{\bf e}\rangle-\langle{\bf e}|K^*|j\rangle&1-\langle k|K|{\bf e}\rangle-\langle{\bf e}|K^*|k\rangle
\end{pmatrix}\\
&=\det\begin{pmatrix}
1-K_{jj}-K_{jj}^*&1-K_{jj}-K_{kk}^*\\1-K_{kk}-K_{jj}^*&1-K_{kk}-K_{kk}^*
\end{pmatrix}\\
&= ( 1-K_{jj}-K_{jj}^* )( 1-K_{kk}-K_{kk}^* )-( 1-K_{jj}-K_{kk}^* )( 1-K_{kk}-K_{jj}^* ) \\
&= (K_{jj}+K_{jj}^*)(K_{kk}+K_{kk}^*)-(K_{jj}+K_{kk}^*)(K_{kk}+K_{jj}^*) \\
&= -K_{jj}K_{jj}^*+K_{jj}K_{kk}^*+K_{jj}^*K_{kk}-K_{kk}K_{kk}^* =-|K_{jj}-K_{kk}|^2
\end{align*}
which implies $K_{jj}=K_{kk}$.
Hence $K=c\cdot{\bf 1}$ for some $c\in\mathbb C$.
All that is left to show is that ${\rm Re}(c)\leq 0$. For this define $\psi:=(1,i)^\top\oplus 0_{n-2}$ and note that $\psi\in\ker(|\psi\rangle\langle\psi|)^\top)$.
Because $\Phi$ is positive one finds
\begin{align*}
0\leq\langle\psi|\Phi(|\psi\rangle\langle\psi|)|\psi\rangle=\langle\psi|(|\psi\rangle\langle\psi|)^\top|\psi\rangle-\langle\psi|K+K^*|\psi\rangle=-4\,{\rm Re}(c)\,,
\end{align*}
that is, ${\rm Re}(c)\leq 0$ and we are done.
%On the other hand, ${\rm tr}(\Phi)=0$ which is equivalent to\footnote{
%Here we use that ${\rm tr}(A(\cdot)B)={\rm tr}(A){\rm tr}(B)$ for all $A,B\in\mathbb C^{n\times n}$ as is readily verified.
%}
%$$
%0={\rm tr}\big(({}^\top-K(\cdot)-(\cdot)K^*\big)=n-n{\rm tr}(K+K^*)=n-4{\rm Re}(K_{11}+K_{nn})
%$$
%but this contradicts ${\rm Re}(K_{11}+K_{nn})\leq 0$ and we are done.
\end{proof}
%this shows: decomposition not possible for \textbf{all} $B$ on $\mathsf P$ (map $\Xi_B$ is not surjective)
%but seems to work for some.
%Explicit example:
%$$
%B=\begin{pmatrix}
%1&2i\\-2i&1
%\end{pmatrix}
%$$
%is Hermitian, ${\rm Re}({\rm tr}(B))\neq 0$. Set $K:=-\frac34\cdot{\bf1}$ and
%$\Phi:={}^\top+\frac32{\rm id}$. Then
%\begin{itemize}
%\item ${\rm Im}({\rm tr}(BK))=0$ $\checkmark$
%\item ${}^\top=K(\cdot)+(\cdot)K^*+\Phi$ $\checkmark$
%\item ${\rm tr}(\Phi(B(\cdot)B)={\rm tr}(B\overline{B})+\frac32{\rm tr}(B)^2 =-6+\frac32*4=0   $ $\checkmark$
%\end{itemize}
%to do:
%\begin{itemize}
%\item does the P example work for general $B$?
%\item ``For all $B\in\mathbb C^{n\times n}$ Hermitian with ${\rm Re}({\rm tr}(B))\neq 0$ one has $$\mathsf{HP}(n)=\{K(\cdot)+(\cdot)K^*:K\in\mathbb C^{n\times n}\}\oplus(\mathsf{CP}_B(n)-\mathsf{CP}_B(n))''$$
%\item 
%\item 
%remark: this shows that all such $\Phi$ have Kraus rank at most $n^2-1$
%$\to$ explains why/in which case $n^2-1$ Lindblad op.s suffice (iff Choi matrix is not full rank iff there exists kernel element $\psi\neq 0$ iff ${\rm tr}(\Phi_{\vecinv\psi})=0$ (?))
%\end{itemize}

In particular, Proposition~\ref{prop_transpos_decomp} rules out decompositions of ${(\cdot)}^\top$ for common classes of matrices $B$ such as, e.g., 
diagonal or real symmetric
%or positive semi-definite\footnote{
%If $B\geq 0$, then $\overline{B}=B^\top\geq 0$ so ${\rm tr}(B\overline{B})={\rm tr}(\sqrt B\,\overline{B}\sqrt B)\geq 0$. Hence the condition from Prop.~\ref{prop_transpos_decomp} (ii) can never be met.
%}
ones.
Examples of matrices $B$ which allow for a decomposition in the sense of Prop.~\ref{prop_transpos_decomp} (i) are $B=|0\rangle\langle 1|$ or---if one wants something positive semi-definite---
$$
B=\begin{pmatrix}
1&i\\-i&1
\end{pmatrix}\,.
$$

\section{Outlook}\label{sec_outlook}
The main contribution of this paper is an uncountable family of unique decompositions of generators of completely positive dynamical semigroups (Theorem~\ref{thm1}).
%This result may even admit a generalization to infinite dimensions, which---beyond partial results \cite{Davies80unique,Evans80,AF83,FM92}---is an open problem.
%What makes generalizing the uniqueness result of Gorini, Kossakowski, and Sudarshan non-trivial is that there is no meaningful counterpart of the trace of \textit{generic} Hamiltonians and Lindblad operators in infinite dimensions.
%Introducing an auxiliary operator $B$ as in Thm.~\ref{thm1} may solve this;
%however, it is not yet clear whether the techniques we used (trace of superoperators, connection to Choi matrix, etc.) can be modified and applied to the infinite-dimensional setting.
%This will be subject to future research.
Moreover we proved that complete positivity is essential for this to hold.

An interesting question---in light of Proposition~\ref{prop_transpos_decomp}---is whether some weaker form of our uniqueness result holds for general \textit{positive} dynamics.
More precisely, does there for all $L\in\mathsf L(\mathsf P(n))$ \textit{exist} some $B$ 
%(probably there is no \textit{universal} $B$ which works for all $\Phi$) 
such that $L$ decomposes uniquely into (some form of) a closed and a dissipative part?
As we saw by the example of the transposition map, beyond complete positivity it becomes non-trivial whether such a decomposition even exists in the first place.
Pursuing this question could also lead to new insights regarding the set of generators of positive dynamical semigroups $\mathsf L(\mathsf P(n))$.
After all, beyond the quantum Kolmogorov conditions \cite[Thm.~5]{Koss72} $\mathsf L(\mathsf P(n))$ has no known closed form like $\mathsf L(\mathsf{CP}(n))$ does.

Moreover, it could be interesting to study the effect of different choices of $B$ on the decomposition of generators in more detail.
We already saw that for such decompositions all Lindblad operators are
%$\langle\cdot,\cdot\rangle_{\sf HS}$-
orthogonal to $B$, and that this decomposition even yields (in some sense) orthogonal components of the generator;
the latter may even be useful for studying symmetries of GKSL-generators, cf.~\cite{OSID17}.
However, while Example~\ref{ex_bloch} gave some insight into this issue for a well-known qubit generator,
connections of these decompositions to notions in quantum information theory (beyond the entanglement fidelity, cf.~Rem.~\ref{rem_entfid}) are yet to be drawn.

%can be used to break symmetries (e.g., covariance symmetry) in a \textbf{controlled} manner. example: bitflip $V=\sigma_x$ ($H=0$). Starting with a general $B$ one sees that the unique decomposition reads
%$$
%H=-{\rm Im}\Big(\frac{b_{12}+b_{21}}{b_{11}+b_{22}}\Big)\sigma_x+\lambda{\bf1}
%$$
%for some $\lambda\in\mathbb R$ (which is, of course, of no consequence for $[H,\cdot]$)
%............

\section*{Acknowledgments}
I would like to thank
%Philippe Faist whose question sparked this paper, to Gunther Dirr for
%making me aware of the paper of Dole\v{z}al \cite{Dolezal64}  (cf.~Remark~\ref{rem_cont_ac_app}, Appendix~\appref{B}), as well as to
%%I am grateful to
%%Gunther Dirr, 
%Emanuel Malvetti and Fereshte Shahbeigi
%%, Thomas Schulte-Herbr\"uggen, and Amit Devra
Sumeet Khatri
%%and the anonymous referee
for his insight that Lemma~\ref{lemma0}---and thus Theorem~\ref{thm1}---holds even beyond Hermitian matrices.
% i and constructive comments during the preparation of this paper.
Moreover I would like to thank Jonas Kitzinger and, again, Sumeet Khatri for proofreading a preliminary version of this manuscript.
This work has been supported by the Einstein Foundation (Einstein Research Unit on Quantum Devices) and the MATH+ Cluster of Excellence.

\bibliographystyle{mystyle}
\bibliography{../../../../../control21vJan20.bib}
\end{document}